\documentclass[11pt]{article}
\usepackage{amsmath}
\usepackage{amsfonts}
\usepackage{latexsym}
\usepackage{amssymb}
\usepackage{graphicx}
\usepackage{amsthm}
\usepackage{chemarrow}
\usepackage{fullpage}
\usepackage{framed}
\usepackage{verbatim}
\usepackage{color}
\usepackage{epsfig}
\usepackage{hyperref}
\usepackage{float}

\title{Catalysis in Reaction Networks}
\author{Manoj Gopalkrishnan\footnote{Part of this work was supported by NSF DMS-0943760}\\
\small School of Technology and Computer Science\\
\small Tata Institute of Fundamental Research\\
\small \texttt{manoj@tcs.tifr.res.in}}
\begin{document}
\maketitle
%\par \noindent

\theoremstyle{definition}
\newtheorem{theorem}{\bf{Theorem}}[section]
\newtheorem{lemma}[theorem]{\bf{Lemma}}
\newtheorem{corollary}[theorem]{\bf{Corollary}}
\newtheorem{proofl}[theorem]{\bf{Proof}}
\newtheorem{open}{\bf{Open}}
\newtheorem{conjecture}{\bf{Conjecture}}

\theoremstyle{definition}
\newtheorem{definition}{\bf{Definition}}[section]

\theoremstyle{remark}
\newtheorem{example}{\bf{Example}}[section]
\newtheorem{notation}[definition]{\bf{Notation}}
\newtheorem{fact}{\bf{Fact}}
\newtheorem{note}[definition]{\bf{Note}}
\newcommand{\CE}{{\cal E}}
\newcommand{\MC}{\mathbb{C}}
\newcommand{\MR}{\mathbb{R}}
\newcommand{\MZ}{\mathbb{Z}}
\newcommand{\MN}{\mathbb{N}}
\newcommand{\MQ}{\mathbb{Q}}
\newcommand{\iton}{i=1,2,\ldots,n}
\newcommand{\jtom}{j=1,2,\ldots, m}
\newcommand{\MP}{\mathbb{R}_{>0}}
\newcommand{\bs}{\boldsymbol}
\newcommand{\op}{\operatorname}

\begin{abstract}
We define \textit{catalytic networks} as chemical reaction networks with an essentially catalytic reaction pathway: one which is ``on'' in the presence of certain catalysts and ``off'' in their absence.  We show that examples of catalytic networks include synthetic DNA molecular circuits that have been shown to perform signal amplification and molecular logic. Recall that a \textit{critical siphon} is a subset of the species in a chemical reaction network whose absence is forward invariant and stoichiometrically compatible with a positive point. Our main theorem is that all weakly-reversible networks with critical siphons are catalytic. Consequently, we obtain new proofs for the persistence of atomic event-systems of Adleman \textit{et al.}, and normal networks of Gnacadja. We define \textit{autocatalytic networks}, and conjecture that a weakly-reversible reaction network has critical siphons if and only if it is autocatalytic.
\end{abstract}

\section{Introduction}
Biological systems exhibit exquisite structure and behavior. We wish to view such sophistication in biological systems as the result of, as well as intended for the performance of, computation. It is hoped that an attitude to seek for algorithms underlying biological computation will give us insight into this sophistication. One of the languages in which the theory of algorithms can be made precise is the language of circuits. What are the circuits, if any, that make up biological systems? Over the past half century, molecular biologists and biochemists have amassed considerable data about biochemical reaction networks. We will postulate that these are the circuits we seek.

Circuits are constituted from the repeated composition of a small number of distinct parts. This circumstance proves of considerable aid in their design and analysis. For example, one may view circuits as made up solely of switches. Indeed, the circuits in one of the first computers, the Z3, were of this type~\cite{Ceruzzi1981Zuse}. What are the switches of biochemical circuits, if such exist?

To answer this question, let us first consider a catalyzed chemical transformation of species $A$ to species $B$, where species $C$ is the catalyst. In principle, if $C$ is left out, $A$ still gets converted to $B$, but at a lower rate. In practise, the lower rate is typically so much lower that the uncatalyzed reaction is left out of the model. This is so especially for biochemical reaction networks --- for example, one does not expect DNA to self-ligate in the absence of a ligase enzyme within the time scales under consideration, even though such a reaction is possible in principle. We submit that catalysts are the switches of biochemical circuits: their presence turns a certain reaction ``on,'' and, for all practical purposes, their absence turns it ``off.'' Intuitively, a reaction network is \textit{catalytic} iff there exists an essentially catalytic reaction pathway: one which is ``on'' in the presence of certain catalysts and ``off'' in their absence. One of the contributions of this paper is to give a mathematical formulation of this notion in Definition~\ref{def:catalytic}.

What dynamical model should we assume for biochemical circuits? There are several choices: mass action, chemical master equation, reaction-diffusion, empirical dynamics, etc. Since very little is known about the true dynamics in a cell, we should only make the most frugal assumptions.

The results proved in this paper hold at the reaction network level. Fixing the network constrains, but does not completely specify, the dynamical model. Therefore, our results are robust across a wide range of dynamical models, including mass action kinetics. This is in the spirit of the approach advocated by Feinberg~\cite{feinberg79lectures}.

Another contribution of this paper is to relate catalysis to the global attractor conjecture, and the persistence conjecture, both of which are long-standing open problems in the theory of chemical reaction networks. The next few paragraphs will explain this connection.

In 1974, Horn~\cite{horn74dynamics} made the global attractor conjecture for complex balanced mass action systems. Complex balance is a condition that restricts specific rate constants to be special enough to guarantee the existence of a ``free energy function.'' Intuitively, the conjecture asserts that positive steady states act like global attractors. That is, every solution trajectory which originates in the positive orthant must asymptotically reach some positive steady state. This conjecture remains open.

One says that a system of differential equations in Euclidean space is \textit{persistent} if no solution trajectory starting in the positive orthant approaches the boundary. Horn~\cite{horn74dynamics} observed that for complex balanced systems, persistence and the global attractor conjecture are equivalent. Note that there exist mass action systems that violate persistence. For example, the reaction $x\rightarrow y$ leads to exponential decay of the concentration of species $x$ to zero.

In 1987, Feinberg~\cite[Remark~6.1.E]{Feinberg19872229} conjectured that persistence must hold if the underlying reaction network is ``weakly-reversible'': i.e., each component is strongly connected. Since all complex balanced systems are weakly-reversible, this conjecture generalizes Horn's conjecture. Very little is known about Feinberg's conjecture: it remains open even when all reactions are reversible and all reaction rate parameters take the value $1$.

In our hands, the notion of ``siphon'' turns out to be the connecting link between catalysis and persistence. A \textit{siphon} is a subset of the reacting species whose absence is forward invariant. A siphon is \textit{critical} if and only if the absence of the siphon species is stoichiometrically compatible with a positive point. These notions were anticipated by
Feinberg~\cite[Proposition~5.3.1, Remark~6.1.E]{Feinberg19872229} without giving them a name. He proved that the zero coordinates of steady state points are siphons, and that in the absence of critical siphons, solutions from the  positive orthant can not asymptotically approach a boundary steady state. Note that this result does not rule out trajectories approaching the boundary from the positive orthant. For example, a solution trajectory originating in the positive orthant may have an omega-limit on the boundary without violating this result.

The term siphon was known in the literature of Petri net theory, and was introduced to reaction networks by Angeli, De Leenheer and Sontag. They proved that the absence of critical siphons implies persistence~\cite[Theorem~2]{Angeli2007598}. Their result is non-trivial and requires delicate arguments about omega-limits.

Anderson and Shiu~\cite{anderson:1464,anderson:1840} have analyzed some types of critical siphons for which persistence can be established. Shiu and Sturmfels~\cite{Shiu2010Siphons} have shown how to compute critical siphons using computer algebra software. Recently, Gnacadja~\cite{Gnacadja2010Reach} has analyzed a related condition called ``vacuous persistence.''

This paper reinterprets critical siphons in the context of biological systems performing computations. We suggest that to implement exponential amplifiers, single-molecule detectors, and molecular logic gates, one requires networks with critical siphons. Our main theorem (Theorem~\ref{thm:norelsiph}) establishes that all weakly-reversible networks with critical siphons are catalytic. One may say that the obstruction to proving Feinberg's persistence conjecture comes from catalytic species that act like switches.

To prove the theorem, we show that the network structure that encodes catalysis, and the network structure that encodes critical siphons, can both be translated into properties of binomial ideals. The crucial insight is that in weakly-reversible networks, the notion of critical siphon does not depend on whether individual reactions are reversible or irreversible (Lemma~\ref{lem:samesiphons}). This allows us to choose all reactions to be reversible, and set all specific rates to be $1$. Following \cite{adleman-2008}, this detailed balanced mass action system can be represented as a set of binomials. The property we seek to investigate manifests itself as a geometric property of the roots of this set of binomials, and verifying this completes the proof.

The proof employs some elementary algebraic geometric ideas (primary decomposition, saturation). This should not be surprising. Indeed, as has been remarked before~\cite{adleman-2008,Craciun20091551}, the theory of reaction networks is intimately connected with the theory of binomials and binomial ideals. Explicitly acknowledging this can sometimes lead to insights, and proof directions, as in the present case.

We introduce the notion of autocatalytic networks (Definition~\ref{def:autocatalytic}), and conjecture that critical siphons occur precisely in autocatalytic networks (Conjecture~\ref{conj:auto}). We analyze an example of a synthetic DNA molecular circuit, the ``seesaw gate''~\cite{seesawgates}, and observe the coincidence of exponential amplification, autocatalysis, and critical siphons (Example~\ref{Seesaw1}).

\section{Preliminaries}
The formal mathematical study of reaction networks was pioneered in the 1970s by Horn, Jackson and Feinberg~\cite{feinberg72complex,feinberg72chemical,feinberg79lectures,horn72necessary,horn74dynamics,horn72general}. More recent overviews include \cite{adleman-2008,Craciun20091551,feinberg95existence,gunawardena2003chemical,sontag01structure}.

Intuitively, a reaction network is a graph. Each node represents a \textit{complex}, or set of chemical species, with each species accompanied by a non-negative integer encoding its multiplicity (or stoichiometry). It is conventional in chemistry to represent complexes in additive notation: e.g., $x + 2y + 5z$ where $x,y,z$ denote chemical species. Following~\cite{adleman-2008} and \cite[Section~1]{Craciun20091551}, we will depart from this convention and use multiplicative notation for complexes: e.g., the complex $x + 2y +5z$ is now written as $xy^2z^5$. Each edge represents a reaction from the source complex to the sink complex.

The multiplicative notation suggests that reversible reactions can be represented by binomials~\cite{adleman-2008}. This representation is key to the notion of ``associated event-system'' (Definition~\ref{def:asses}).

Succinctly, a reaction network is a graph where the nodes are labeled by monomials in the species. A reaction network is called \textit{weakly-reversible} if, whenever there is a path (sequence of directed edges) from one complex to another, there is a path back. To assign a dynamics to the network, we need to assign a \textit{specific rate constant} to each reaction: this is the rate at which the reaction would take place if all species had unit concentration. With such an assignment, the reaction network is called a \textit{mass action system}. We recall the definitions, following~\cite[Section~1]{Craciun20091551}.

\begin{definition}\
\begin{enumerate}
\item A ``chemical reaction network'' (CRN) consists of the following data:
\begin{enumerate}
\item Positive integers $s$ (for number of species) and $n$ (for number of complexes),
\item A finite directed graph $G$ with vertices $V(G) = \{1,2,\cdots,n\}$, and edges $E(G)\subseteq V\times V$, and
\item An $s$-variable monomial labeling of the vertices \[\psi(x)=\left(\psi_1(x),\psi_2(x),\cdots,\psi_n(x)\right) = \left(\prod_{j=1}^s x_j^{y_{1j}}, \prod_{j=1}^s x_j^{y_{2j}},\cdots,\prod_{j=1}^s x_j^{y_{nj}} \right)\] such that if $\psi_i=\psi_{i'}$ then $i=i'$.
\end{enumerate}
\item A CRN is ``weakly-reversible'' iff each connected component is strongly-connected.
\item A ``mass action system'' (MAS) is a CRN with a weight function $k:E(G)\rightarrow \MR_{>0}$.
\end{enumerate}
\end{definition}

\begin{notation}\
\begin{enumerate}
\item The matrix $A_{k}$ denotes the negative of the Laplacian of $G$. That is, when $i\neq j$, the entry in position $(i,j)$ is $k(i,j)$, and the row sums are zero.
\item The matrix of non-negative integers $(y_{ij})_{n \times s}$ is denoted by $Y = Y_{\psi}$. The $i^{\text{th}}$ row $(y_{i1},y_{i2},\cdots,y_{is})$ is denoted by $y_i$.
\item The tuple $(x_1,x_2,\cdots,x_s)$ is denoted by $x$.
\item For all positive integers $m$, $[m]$ denotes the set $\{1,2,\cdots,m\}$.
\item The set $\mathcal{M}_s = \left\{ \prod_{i=1}^s x_i^{a_i} \mid a\in\MZ_{\geq 0}^s\right\}$ denotes all monic monomials in $s$ variables.
\item The graph $U_G$ is the underlying undirected graph of $G$ with nodes $V(G)$ and edges $\{(i,j)\mid $ either $(i,j)\in E(G)$ or $(j,i)\in E(G)\}$.
\end{enumerate}
\end{notation}

\begin{definition}[Mass Action Kinetics]
The \textit{associated dynamical system} of an MAS $\langle G,\psi,k\rangle$  is given by:
\begin{align*}
\frac{dx(t)}{dt} = \psi(x(t))\cdot A_k \cdot Y.
\end{align*}
\end{definition}

We will sometimes write only the graph for CRN's and MAS's. A reference to $G$ could mean either the graph $G$, or the CRN $\langle G,\psi\rangle$, or the MAS $\langle G,\psi, k\rangle$. The meaning should be clear from context.

\begin{definition}
The \textit{stoichiometric subspace} $S=S_G$ of a CRN $G$ is the linear subspace of $\MR^s$ spanned by $\{ y_{i} - y_{i'} \mid (i,i')\in E(G)\}$, the row-differences of $Y$.
\end{definition}

For all $x$, $\psi(x)\cdot A_k \cdot Y$ lies within $S$. Therefore, all solutions to the associated dynamical system of an MAS lie within affine translates of $S$. Moreover, if the solution originates at a non-negative point (i.e., $x(0)\in\MR^s_{\geq 0}$), then it is known that the solution trajectory is confined to the polyhedron obtained by intersecting the appropriate affine translate of the stoichiometric subspace with the non-negative orthant~\cite[Theorem~4.5]{adleman-2008},~\cite[Corollary~7.3]{sontag01structure}. Following \cite{Craciun20091551}, this polyhedron will be called the ``invariant polyhedron.''

\begin{definition}
Let $G$ be a CRN. For all $x\in\MR^s_{\geq 0}$, the \textit{invariant polyhedron} containing $x$ is $P_x = (x+S)\cap\MR^s_{\geq 0}$.
\end{definition}

Note that invariant polyhedra do not necessarily intersect the positive orthant $\MR^s_{>0}$.

Note that $1$ is a monic monomial. Given two monic monomials $M=\prod_{i=1}^s x_i^{e_i}$ and $N=\prod_{i=1}^s x_i^{f_i}$ from $\mathcal{M}_s$, recall~\cite{adleman-2008} that $M$ {\em precedes} $N$ (and we write $M\prec N$) iff $M \neq N$ and for the least $i$ such that $e_i \not = f_i$, we have $e_i<f_i$.

\begin{definition}[Associated Event-System]\label{def:asses}
Let $G$ be a weakly-reversible CRN. The ``associated event-system'' $\CE_G$ of $G$ is the set of binomials $\{\psi_i - \psi_j \mid \psi_i\prec\psi_j$ and either $(i,j)$ belongs to $E(U_G)$ or $(j,i)$ belongs to $E(U,G)\}$.
\end{definition}

The monomial order $\prec$ chosen in \cite{adleman-2008} was arbitrary; any total order on monomials would do to ensure that the map from a CRN to its associated event-system is well-defined.

Note that $\CE_G$ is a ``natural event-system'' in the language of Adleman \textit{et al.}~\cite[Definition~2.8]{adleman-2008}. It carries the same information as the MAS $\langle U_G, 1\rangle$, but this information is expressed in terms of polynomials. This is because of the one-to-one correspondence between binomials in $\CE_G$ and edges in $U_G$.

Recall that a pure difference binomial is a binomial of the form $M - N$ where $M,N\in \mathcal{M}_s$ are monic monomials. An ideal generated by pure difference binomials will be called a pure difference binomial ideal. Note that when $G$ is a weakly-reversible CRN, $(\CE_G)$ is a pure difference binomial ideal.

We will consider the polynomial ring $\mathbf{k}[x] = \mathbf{k}[x_1,x_2,\cdots,x_s]$ where $\mathbf{k}$ is a field of characteristic zero. $(\CE_G)$ will denote the ideal generated by $\CE_G$. Note that the ideal $\mathfrak{J}_G\subseteq \MQ[x]/(x_1x_2\cdots x_s)$ of Shiu and Sturmfels~\cite{Shiu2010Siphons} is the quotient of $(\CE_G)$ when $\mathbf{k}=\MQ$.

\begin{example}
Consider the MAS $\langle G,\psi,k\rangle$:
\begin{figure}[H]
    \centering
    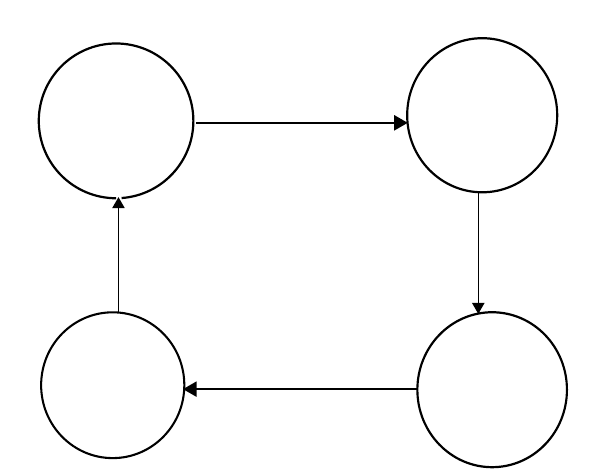
\end{figure}
The number of species is $s=2$. The number of complexes is $n=4$. Each complex is labeled with a monic monomial in $x$ and $y$. The CRN is weakly-reversible, since the only component is strongly-connected. The associated dynamical system is given by:
\begin{align*}
\left(\dot{x}, \dot{y}\right) &= (xy^2, x^4, y^3, 1) \left( \begin{array}{cccc}
-k_{1,2} & k_{1,2} & 0 & 0 \\
0 & -k_{2,3} & k_{2,3} & 0 \\
0 & 0 & -k_{3,4} & k_{3,4} \\
k_{4,1} & 0 & 0 & -k_{4,1}
\end{array} \right)\left( \begin{array}{cccc}
1 & 2\\
4 & 0 \\
0 & 3  \\
0 & 0
\end{array} \right)\
\\&=(3 k_{1,2} xy^2 - 4k_{2,3} x^4 + k_{4,1}, -2 k_{1,2} xy^2 + 3 k_{2,3} x^4 - k_{3,4} y^3 + 2 k_{4,1}).
\end{align*} The stoichiometric subspace $S_G$ is $\MR^2$, and the invariant polyhedron is $\MR^2_{\geq 0}$. The associated event-system $\CE_G$ is $\{xy^2 - x^4, y^3 - x^4, 1-y^3, 1-xy^2\}$.
\end{example}

\section{Catalysis and Catalytic Networks}
The remarkable catalytic behavior of enzymes is one of the most striking features of biochemical networks, and arguably central to life itself. Definition~\ref{def:catalytic} makes precise what it means for a network to behave in a ``catalytic'' manner. The next three examples motivate Definition~\ref{def:catalytic}.

For a single reaction, one says that a catalyst is a species whose availability changes the rate of a reaction but which is left unchanged by the reaction.

\begin{example}\label{Example1}
For the reversible reaction $\{ x + y \autorightleftharpoons{}{} x + z\}$, $x$ catalyzes the conversion of $y$ to $z$.
\end{example}

The above example suggests the following rule for identifying catalysts in a system of reactions: a species $x_i$ is a catalyst for the CRN $G$ if and only if $x_i$ occurs on both sides of some reaction.

\begin{example}\label{Example2}
Consider the reactions
\begin{align*}
    x+y &\autorightleftharpoons{}{} p\
    \\p+q &\autorightleftharpoons{}{} x + z
\end{align*}
There is no species that occurs on both sides of some reaction. Hence, by the proposed rule, one would conclude that there are no catalytic species. However, consider the following sequence of reactions: $x$ combines with $y$ to form $p$, $p$ combines with $q$ to form an $x$ and a $z$. The net result is that $x$, $y$ and $q$ combine to form an $x$ and a $z$. The availability of $x$ changes the rate of this reaction pathway, but $x$ is left unchanged by the pathway. Hence, it appears desirable to call species $x$ a catalyst.
\end{example}

From Example~\ref{Example2}, we see that in the presence of multiple reactions, the concept of catalyst requires a little more care in definition. We could postulate that for a network of reactions, a catalyst is a species whose availability changes the rate of some reaction pathway, but which is left unchanged by that reaction pathway, and attempt to make this precise. Instead, let us turn our attention to a related question concerning the entire network.

It is a common expectation in chemistry that a chemical transformation which takes place in the presence of a catalyst must also take place in its absence, though perhaps at a much slower rate. However, most biochemical networks of interest violate this expectation. For example, consider the ligation of a nicked DNA double strand. In the absence of DNA ligase, this expectation leads us to believe that the ligation still takes place. However, in practise, the rate is so low that this reaction is never explicitly included in the network.

Networks that have such ``essentially catalytic'' pathways we will call ``catalytic'' networks. When framing the formal definition, it is our intention that the networks of Examples~\ref{Example1} and~\ref{Example2} be catalytic because, in both cases, there exist reaction pathways on which species $x$ acts like a switch.

Consider a biochemical network that models the ligation of nicked DNA double strands. Suppose the very slow reactions corresponding to the spontaneous self-ligation of the DNA double strands in the absence of ligase are also included in the network. One expects kinetic simulations to be at least as accurate as when these very slow reactions are left out. On the other hand, steady-state analysis of such an enhanced network may be much less accurate. This may seem surprising at first, for how can a more accurate network model lead to a less accurate prediction? This apparent paradox is really an issue of time scales. If it takes an astronomical amount of time to approach steady state, then the steady state prediction, though perfectly valid at some time scale, is of doubtful utility. Therefore, for the purposes of steady-state analysis, it is something of an art to leave out reactions that are not expected to play a substantial role within time scales of interest. The networks that omit such very slow reactions are likely to be catalytic, especially in the context of biochemistry, because the enzymes that speeds up these very slow reactions acts like switches.

One more nuance needs to factor into our definition of ``catalytic,'' as we now illustrate.

\begin{example}\label{Example3}
Consider the reactions
\begin{align*}
x &\autorightleftharpoons{}{} p\
  \\y &\autorightleftharpoons{}{} q\
  \\x+y+w &\autorightleftharpoons{}{} p+q+w
\end{align*}
We informally identify the species $w$ as a catalyst because its availability changes the rate of the reaction $x+y+w \autorightleftharpoons{}{} p+q+w$, and $w$ is left unchanged by the reaction. There is no reaction of the form $x+y\autorightleftharpoons{}{} p+q$. Therefore, one may be tempted to say that this network is catalytic, and that the species $w$ acts as a switch on the reaction pathway from $x+y$ to $p+q$. However, for $w$ to be essential to this pathway, we intend that the absence of $w$ isolate the two ends of the pathway. In particular, suppose $x$ and $y$ have positive concentrations, and $p$ and $q$ have zero concentration. Then the concentrations of $p$ and $q$ should rise off zero if, and only if, $w$ is present. For this network, the parallel pathways $x \autorightleftharpoons{}{} p$ and $y \autorightleftharpoons{}{} q$ together provide a ``leakage current'' which make $w$ non-essential. We will define ``catalytic'' in such a way as to exclude networks like this one.
\end{example}

Note that in Example~\ref{Example3}, no vertex in the corresponding CRN is labeled with either of the monomials $xy$ or $pq$. Hence, whether a network is catalytic depends on monomials that may not appear as labels of vertices in $G$. We will extend the graph $G$ to a graph $\overline{G}$ that includes every monic monomial as a vertex, and extends the edges of $G$ in the natural manner. We recall the definition of an event-graph from \cite[Definition 2.9]{adleman-2008}. For our present purposes, the following simplified definition suffices.

\begin{definition}[Event-graph]\label{def:evtgraph}
Let $G$ be a weakly-reversible CRN. The \textit{event-graph} $\overline{G}$ of $G$ is a directed graph  with vertices $V(\overline{G})=\mathcal{M}_s$, and edges $(N \psi_i, N \psi_j)$ where $(i,j)\in E(G)$ and $N\in \mathcal{M}_s$ is a monic monomial.
\end{definition}

Note that if $G$ is weakly-reversible then every connected component of $\overline{G}$ is strongly-connected. Further, $G$ can be viewed as a subgraph of $\overline{G}$ in the natural manner.

\begin{example}
Let $G$ be the network in Example~\ref{Example2}. The event-graph $\overline{G}$ has vertices all monomials in the variables $x,y,p,q$, and $z$. The edges are the obvious edges induced by the network. For example, from the monomial $xyq$, by applying the first reaction, there is an edge to the monomial $pq$. From $pq$, by applying the second reaction, there is an edge to the monomial $xz$.
\end{example}

\begin{definition}[Catalytic]\label{def:catalytic}
A weakly-reversible CRN $G$ is \textit{catalytic} iff there exist path-connected $M,N\in V(\overline{G})$ such that $M/\gcd(M,N)$ and $N/\gcd(M,N)$ are not path-connected in $V(\overline{G})$.
\end{definition}

\begin{example}
Let $G$ be the network in Example~\ref{Example2}. The monomials $xyq$ and $xz$ are path-connected in $\overline{G}$ by the path $xyq \rightarrow pq \rightarrow xz$. Dividing by $\gcd(xyq,xz) = x$ yields the monomials $yq$ and $z$ which are not path-connected in $\overline{G}$. This last assertion is true because $z$ is an isolated node of $\overline{G}$, since there is no complex in $G$ that divides $z$. Therefore, $G$ is catalytic.
\end{example}

Note that in this definition, $\gcd(M,N)$ plays the role of a catalyst. On similar lines, it is easy to verify that the CRN in Example~\ref{Example1} is catalytic, and that in Example~\ref{Example3} is not, as intended.

Lemma~\ref{lem:saturated} will provide an equivalent algebraic characterization of ``catalytic'' in terms of saturated ideals. Motivated by this connection, a previous draft of this paper referred to weakly-reversible CRN's that are not catalytic as ``saturated.''

The content of the next lemma is the link between the ideal $(\CE_G)$ generated by the associated event-system in $\mathbf{k}[x]$, and the event-graph $\overline{G}$.

\begin{lemma}\label{lem:evtgraph}
Let $G$ be a weakly-reversible CRN. Let $M$ and $N$ be distinct monic monomials. Then $M-N \in (\CE_G)$ iff $M$ and $N$ are path-connected in $\overline{G}$.
\end{lemma}
\begin{proof}
($\Leftarrow$) Consider a path $M=M_1,M_2,\cdots,M_l = N$ in $\overline{G}$ from $M$ to $N$. Then for $i=1$ to $l-1$, $M_i - M_{i+1}$ belongs to $(\CE_G)$. Adding, we get $M - N \in (\CE_G)$.\
\\($\Rightarrow$) Suppose $M-N\in (\CE_G)$. There exist $c_{ijk}\in \mathbf{k}$ and monic monomials $M_{ijk}$ and a positive integer $K$ such that:
\begin{align}\label{eqn:sum}
M - N = \sum_{(i,j)\in E(G)}\left(\sum_{k=1}^K c_{ijk} \left(M_{ijk}\cdot\psi_j-M_{ijk}\cdot\psi_i\right)\right).
\end{align}
Suppose for the sake of contradiction that $M$, $N$ are not path-connected in $\overline{G}$. Then $M$ and $N$ are vertices in distinct components of $\overline{G}$, say $K_1$ and $K_2$. Restricting Equation~\ref{eqn:sum} to monomials in $V(K_1)$, and noting that for all $i,j,k$, the monomials $M_{ijk}\cdot\psi_j$ and $M_{ijk}\cdot\psi_i$ are path-connected in $\overline{G}$, we have:
\[M = \sum_{M_{ijk}\cdot\psi_j\in V(K_1)} c_{ijk} \left(M_{ijk}\cdot\psi_j-M_{ijk}\cdot\psi_i\right).\]
Evaluating this expression at the point $\langle 1,1,\cdots,1\rangle$, we get $1=0$, a contradiction.
\end{proof}

The next lemma translates the notion of catalytic networks into algebraic terms. Recall from~\cite[Exercise~4.4.8]{Cox1991Ideals} that the saturation of an ideal $I\subseteq \mathbf{k}[x_1,x_2,\cdots,x_s]$ with respect to a polynomial $f$ is the ideal $I:f^\infty = \{g\in \mathbf{k}[x_1,x_2,\cdots,x_s]\mid$ there exists $k>0$ with $f^kg\in I\}$.

\begin{lemma}\label{lem:saturated}
A weakly-reversible CRN $G$ is catalytic iff $(\CE_G) \subsetneq (\CE_G):(x_1x_2\cdots x_s)^{\infty}$.
\end{lemma}
\begin{proof}
Note that the saturation $(\CE_G):(x_1x_2\cdots x_s)^{\infty}$ of the binomial ideal $(\CE)$ is itself a binomial ideal by \cite[Corollary~1.7.(a)]{eisenbud1994binomial}. Therefore $(\CE_G):(x_1x_2\cdots x_s)^{\infty}$ is generated by the saturates $\{\frac{N-M}{\gcd{(M,N)}}\mid (M,N) \in E(\overline{G})\}$ of binomials.\
\\$(\Leftarrow)$ If $G$ is not catalytic then by Lemma~\ref{lem:evtgraph} all binomials $\frac{N-M}{\gcd(M,N)}$ are in $(\CE_G)$.\
\\$(\Rightarrow)$ Suppose $G$ is catalytic. Then there exist distinct, path-connected $M_1,M_2\in V(\overline{G})$ such that $M_2/\gcd(M_1,M_2)$ to $M_1/\gcd(M_1,M_2)$ are not path-connected in $V(\overline{G}$. By Lemma~\ref{lem:evtgraph}, $M_2/\gcd(M_1,M_2) - M_1/\gcd(M_1,M_2) \in (\CE_G):(x_1x_2\cdots x_s)^{\infty}\setminus(\CE_G)$.
\end{proof}

\section{Networks with Critical Siphons are Catalytic}
Recall from \cite{Angeli2007598} the concept of ``persistence'' of positive dynamical systems, which models the idea of ``species non-extinction.'' This notion can be extended to chemical reaction networks in a natural manner: a CRN is persistent if, for every choice of specific rate constants, the corresponding MAS is persistent. Recall that for $x:\MR_{\geq 0}\rightarrow\MR^s$, the omega-limit set $\omega(x)$ of $x$ is the set $\{y\in\MR^s\mid$ there exists an increasing sequence $t_1,t_2,\dots$ tending to infinity such that the sequence $x(t_1), x(t_2),\dots$ is Cauchy with limit $y\}$. In words, the omega-limit set of a trajectory consists of those points that are approached arbitrarily close, infinitely often, along some unbounded increasing sequence of times.

%Recall for a dynamical system on a metric space $X$, the notion of omega-limit set of $x\in X$, or $\omega(x)$.

\begin{definition}[Persistence]
A CRN $\langle G,\psi\rangle$ is \textit{persistent} iff for every MAS $\langle G,\psi,k\rangle$, for every solution $x:\MR_{\geq 0}\rightarrow\MR^s_{>0}$ to the associated dynamical system, the omega-limit set of $x$ does not meet the boundary of the positive orthant. That is, $\omega(x)\cap \partial\MR^s_{\geq 0} = \emptyset$.
\end{definition}

Note that by our definition a persistent MAS is allowed solutions that escape to infinity, as well as ``omega-limit points at infinity'' where some coordinates are infinite and others may be zero, since such points at infinity do not belong to $\MR^s_{\geq 0}$, but to a compactification.

\begin{open}[Feinberg's Persistence Conjecture, 1987]~\cite[Remark~6.1.E]{Feinberg19872229}\label{open:feinberg}
All weakly-reversible CRNs are persistent.
\end{open}

One of the consequences of Feinberg's conjecture is Horn's global attractor conjecture which has remained open since 1974. Recall that an MAS $\langle G, \psi, k\rangle$ is \textit{complex balanced} iff there exists a positive point $x_0\in\MR^s_{>0}$ such that $\psi(x_0)\cdot A_k = 0$.

\begin{open}[Horn's Global Attractor Conjecture, 1974]\cite{horn74dynamics}
Let $G$ be a complex balanced MAS. Then every invariant polyhedron $P$ contains a positive point $x^*$ such that for all solutions $x(t)$ to the associated dynamical system with $x(0)\in P\cap\MR^s_{>0}$, the limit $\displaystyle\lim_{t\rightarrow \infty} x(t)$ exists, and equals $x^*$.
\end{open}

Recall that an MAS is \textit{detailed balanced} iff all reactions are reversible, and the specific rate constants are such that the reactions admit a positive point of simultaneous balance of all the reactions. Every detailed balanced system is also complex balanced. Even for detailed balanced systems with every specific rate constant equal to $1$, the persistence and global attractor conjectures are open problems.

Recall the notions of siphon and critical siphon from~\cite{Angeli2007598}. The intuitive idea is that a siphon is a set of species whose absence is forward invariant. A siphon is critical if the absence of the siphon species is stoichiometrically compatible with a positive point. The important point to note is that this definition can be made at the level of the network, without making any reference to the dynamics.

\begin{definition}[Siphon, critical siphon]\label{def:siphon}
Let $G$ be a CRN. A nonempty set $Z\subseteq [s]$ is a \textit{siphon} of $G$ iff for all $(i,j)\in E(G)$, if there exists $k\in Z$ such that $x_k \mid \psi_j$ then there exists $l\in Z$ such that $x_l \mid \psi_i$. A siphon $Z$ is \textit{critical} iff there exists a point $z\in\MR^s_{\geq 0}$ such that $Z=\{i\mid z_i=0\}$ and the invariant polyhedron containing $z$ intersects $\MR^s_{>0}$.
\end{definition}

\begin{example}
For the network of Example~\ref{Example1}, every siphon containing (the index of) $z$ must also contain either $x$ or $y$. However, the absence of $x$ is forward-invariant, so the set $\{x\}$ is itself a siphon.
\end{example}

We will now reinterpret critical siphons in the context of biological systems performing computations. Many circumstances in biology require biochemical circuits to detect extremely small concentrations of some species, or the simultaneous presence of several species. This is a signal amplification task, and at the same time a logical AND operation. Let us call an operation that performs both these tasks simultaneously an ``amplifying AND'' operation. For example, exponential amplifiers, single-molecule detectors which must report on the presence or absence of small concentrations of some species, and molecular logic gates, all require ``amplifying AND'' operations.

For these circuits, the presence of input should lead to amplification of one or more ``reporter species,'' preferably exponentially fast so as to exceed a critical threshold of detection within a reasonable amount of time. We call this the \textit{sensitivity} requirement for the amplifying AND operation. This sensitivity requirement appears to be important for the error-free operation of large circuits.

In order to ensure that there are no false positives, under appropriate initial conditions, absence of input and reporter species must be forward invariant in time. This is a \textit{robustness} requirement. To ensure robustness of the amplifying AND gate, some superset of the reporter species together with the inputs (where the extra species act as control) must correspond to a siphon.

Exponential amplification requires that small and large concentrations of the reporter species be stoichiometrically compatible. We suggest that this must lead to a stoichiometrically compatible direction transverse to the siphon, and that critical siphons are related to exponential amplification. The example below provides evidence in favor of this suggestion.

\begin{example}\label{Seesaw1}
The ``seesaw gate'' of Qian and Winfree~\cite{seesawgates} is a synthetic DNA network primitive analogous to a transistor. Like transistors, seesaw gates are intended to be composed to form large circuits that can perform a variety of functions like amplification, digital logic, etc. Zhang \textit{et al.}~\cite{Zhang16112007} have demonstrated exponential amplification with a related motif.%, and Qian and Winfree~\cite{seesawgates2} have reported initial experimental success in building digital logic circuits using seesaw gates.

Here we consider the CRN of a one input, one output seesaw gate, modeled by the two reversible reactions below, involving the five chemical species $gate:output$, $input$, $gate:input$, $output$, and $gate:fuel$.
\begin{align*}
gate:output + input &\autorightleftharpoons{}{} gate:input + output\
\\gate:input + fuel &\autorightleftharpoons{}{} gate:fuel + input
\end{align*}

It is clear that $\{input, gate:input\}$ is a siphon because the absence of $input$ and $gate:input$ is forward invariant. This confers some robustness to a seesaw gate. One wishes to have even more robustness against small amounts of noise, since empirical observation suggests that small leakages are inevitable when composing larger circuits. Ideally, the output should be released only when the input exceeds some threshold. To accomplish this, a small concentration of an additional ``threshold'' species that absorbs ``input'' species is introduced. This new reaction \[input + threshold\text{ } \autorightarrow{}{} \text{ }waste\] provides negative feedback to small leakages until the threshold species is exhausted. One may say that this negative feedback ``thickens the siphon.''

Note that the CRN is catalytic. There are edges in the event-graph from the complex $fuel + gate:output + input$ to the complex $fuel + gate:input + output$, to $output + gate:fuel + input$, but there is no path from $fuel + gate:output$ to $output + gate:fuel$. Thus, $input$ acts as a catalyst.

The siphon $\{input, gate:input\}$ is not critical because the sum of concentrations of $input$ and $gate:input$ is a dynamical invariant, \textit{i.e.}, its value remains unchanged along solution trajectories. Correspondingly, one may verify that the seesaw gate is not an exponential amplifier.

If we set output and input to be the same species, then the net reaction becomes \[fuel + gate:output + input \autorightleftharpoons{}{} gate:fuel + 2\text{ } input.\] This leads to exponential amplification, as verified by simulation by Qian and Winfree~\cite{seesawgates}. Correspondingly, $\{input, gate:input\}$ is now a critical siphon.

For this last reaction network, note that the catalyst $input$ is able to catalyze its own production. We will see later (Definition~\ref{def:autocatalytic}) that this is an example of an ``autocatalytic network.''
\end{example}

We have suggested that critical siphons may be important for achieving interesting behavior in biochemical circuits. Our main theorem says that:

\begin{theorem}\label{thm:norelsiph}
Weakly-reversible networks with critical siphons are catalytic.
\end{theorem}

The proof of Theorem~\ref{thm:norelsiph} appears at the end of this section. There are many examples of CRNs that have siphons but are non-catalytic. For example, any non-catalytic network that is conservative (conserves mass) has a siphon consisting of all the species. Therefore, the adjective ``critical'' is required in Theorem~\ref{thm:norelsiph}.

Theorem~\ref{thm:norelsiph} can be used to prove the persistence of non-catalytic, weakly-reversible networks. This follows from a non-trivial result of Angeli, De Leenheer, and Sontag, who show that reaction networks without critical siphons are persistent~\cite[Theorem~2]{Angeli2007598}. We sketch their proof idea. They prove the contrapositive by considering a solution trajectory $x(t)$ that originates in the positive orthant ($x(0)\in\MR^s_{>0}$) and whose omega-limit set $\omega$ intersects the boundary, and produce from $\omega$ a critical siphon. They consider the set $S$ which is the intersection of this omega-limit set $\omega$ with the boundary $\partial \MR^s_{\geq 0}$, and claim it is an invariant set. If not forward invariant, then there is a trajectory starting on $S$ that enters the positive orthant. Following this trajectory backwards in time, one observes that it leaves the non-negative orthant. This is a contradiction, since omega-limit sets are invariant. Similarly, we obtain backward invariance. It remains to observe that from this invariant set $S$ one can derive a critical siphon, which essentially follows from definitions.

Though they state their theorem only for conservative networks (where trajectories are bounded), their proof nowhere makes use of this assumption. Their result requires neither weak-reversibility nor complex balancing.

Therefore, as a corollary of Theorem~\ref{thm:norelsiph}, we have:

\begin{corollary}\label{cor:persistent}
Every weakly-reversible, non-catalytic CRN is persistent.
\end{corollary}

In fact, Corollary~\ref{cor:persistent} holds for a larger class of dynamical models than mass action kinetics. See \cite{Angeli2007598} for details. By Corollary~\ref{cor:persistent}, weakly-reversible CRNs whose persistence remains unknown must be catalytic. Therefore, one may say that the obstruction to proving Feinberg's persistence conjecture comes from catalysis.

Complex balanced systems are weakly-reversible. Therefore, by Corollary~\ref{cor:persistent}, non-catalytic complex balanced systems are persistent. Since the global attractor conjecture is true for persistent, complex balanced systems~\cite{horn74dynamics}, we have:

\begin{corollary}\label{cor:gas}
The global attractor conjecture is true for non-catalytic, complex balanced MASs.
\end{corollary}

We now establish a number of lemmas that will be required for the proof of Theorem~\ref{thm:norelsiph}. The first task is to translate the notion of critical siphon to algebra. We first translate from the CRN $G$ to the underlying undirected CRN $U_G$ which corresponds to setting every reaction to be reversible, and from there to the associated event-system $\CE_{G}$, which corresponds to setting all reaction rates to have the value $1$.

\begin{lemma}\label{lem:samesiphons}
Let $G$ be a weakly-reversible CRN. The CRNs $U_G$ and $G$ have the same siphons, and critical siphons.
\end{lemma}
\begin{proof}
Let $Z\subseteq [s]$. If $Z$ is a siphon of $U_G$ then it is immediate that $Z$ is a siphon of $G$.

Suppose $Z$ is a siphon of $G$. We claim that $Z$ is a siphon of $U_G$. Consider $(i,j)\in E(G)$. We need to prove that if there exists $k\in Z$ such that $x_k \mid \psi_i$ then there exists $l\in Z$ such that $x_l \mid \psi_j$. Suppose $x_k \mid \psi_i$. Since $E(G)$ is weakly-reversible, there is a path $\langle j = i_0, i_1,\cdots,i_p = i\rangle$ from $j$ to $i$ in $E(G)$. Since $Z$ is a $G$-siphon, there exists $l_{p-1}\in Z$ such that $x_{l_{p-1}} \mid \psi_{i_{p-1}}$. By repeating this argument backwards along the path till we reach the node $j$, the claim follows.

Since $G$ and $U_G$ have the same invariant polyhedra, it is immediate that a siphon $Z$ is critical for $G$ iff it is critical for $U_G$.
\end{proof}

Note that, by~\cite[Theorem~3.1]{Shiu2010Siphons}, inclusion-minimal siphons of $G$ could have been defined directly in terms of the ideal $(\CE_G)$. This would give an alternate proof of Lemma~\ref{lem:samesiphons}.

The critical siphons of $U_G$ can be studied by studying the zeros of $\CE_G$, as we now show. First we need a lemma from linear algebra.

\begin{lemma}\label{lem:linal}
Let $s$ be a positive integer, $S$ be a vector subspace of $\MR^s$, $b\in\MR^s_{>0}$ be a positive point and $z \in S+b$ be a non-negative point. Consider $\alpha\in \{0,1\}^s$ with coordinates $\alpha_i = 1$ iff $z_i>0$, else $\alpha_i = 0$. Then the invariant polyhedron containing $\alpha$ intersects $\MR^s_{>0}$.
\end{lemma}
\begin{proof}
Let $\lambda\in\MR_{>0}$ be such that for $\iton$, $|\lambda(b_i-z_1)|<1$. Consider $d=\alpha + \lambda (b-z)$. Since $\lambda(b-z)\in S$, we have $\alpha\in S+d$. It is enough to show that $d\in\MR^s_{>0}$. For all $\iton$, either\
\\Case 1: $\alpha_i=0$, in which case $z_i=0$, and so $d_i = \lambda(b_i -  z_i) = \lambda b_i>0$, or\
\\Case 2: $\alpha_i=1$, in which case $|\lambda(b_i-z_i)|<1$ implies $d_i>0$.
\end{proof}

When $\mathbf{k}=\MR$ or $\mathbf{k}=\MC$, and $\CE\subseteq\mathbf{k}[x]$, the variety $V_\mathbf{k}(\CE)$ will denote the zeros of $\CE$ over $\mathbf{k}$ with the analytic topology. So far, we have been talking about siphons as subsets of $[s]$. We will find it useful to talk about points whose zero sets are siphons. We call such points siphon-points.

\begin{definition}[Siphon-point, critical siphon-point]
Let $G$ be a weakly-reversible CRN, and $Z\subseteq [s]$ be non-empty. A point $\alpha\in \MR^s_{\geq 0}$ is a \textit{$Z$-siphon-point} iff $\alpha\in V_{\MR}(\CE_G)$ and $Z=\{i\mid \alpha_i = 0\}$. The $Z$-siphon point $\alpha$ is a \textit{critical $Z$-siphon-point} iff the invariant polyhedron containing $\alpha$ intersects $\MR^s_{>0}$.
\end{definition}

\begin{lemma}\label{lem:relsiph}
Let $G$ be a weakly-reversible CRN. $Z\subseteq [s]$ is a siphon of $G$ iff there exists a $Z$-siphon-point $\alpha\in \MR^s_{\geq 0}$. Moreover, $Z$ is a critical siphon iff there exists a critical $Z$-siphon-point $\alpha$.
\end{lemma}
\begin{proof}
($\Rightarrow$) Suppose $Z$ is a siphon of $G$. Consider the point $\alpha\in \{0,1\}^s$ with coordinates $\alpha_i = 0$ iff $i\in Z$, else $\alpha_i = 1$. Since $Z$ is a $U_G$-siphon by Lemma~\ref{lem:samesiphons}, it is easily verified that $\alpha\in V_{\MR}(\CE_G)$.

Now suppose $Z$ is critical. From Definition~\ref{def:siphon}, there exists a point $z\in\MR^s_{\geq 0}$ such that $Z=\{i\mid z_i=0\}$ and the invariant polyhedron containing $z$ intersects $\MR^s_{>0}$. From Lemma~\ref{lem:linal} with $S=S_G$, the invariant polyhedron containing $\alpha$ intersects $\MR^s_{>0}$.

($\Leftarrow$) Suppose $\alpha\in V_{\MR}(\CE_G)\cap \MR^s_{\geq 0}$ is such that $Z=\{i\mid \alpha_i = 0\}$. Then it is immediate that $Z$ is a $U_G$-siphon. From Lemma~\ref{lem:samesiphons}, $Z$ is a $G$-siphon.

Further, if the invariant polyhedron containing $\alpha$ intersects $\MR^s_{>0}$ then it is immediate that $Z$ is critical.
\end{proof}

Our task for proving Theorem~\ref{thm:norelsiph} is now reduced to showing that when a weakly-reversible CRN $G$ is not catalytic, $V_{\MR}(\CE_G)$ contains no critical siphon-points. By ``strongly dense'' we will mean dense in the analytic topology.

\begin{lemma}\label{lem:dense}
Let $I\subseteq\MC[x_1,x_2,\cdots,x_s]$ be an ideal such that $I = I : (x_1x_2\cdots x_s)^{\infty}$. Then $V_{\MC} (I)\setminus V_{\MC}(x_1x_2\cdots x_s)$ is strongly dense in $V_{\MC}(I)$.
\end{lemma}
\begin{proof}
From~\cite[Theorem~VIII.3.5]{Hungerford1974Algebra}, there exists $k\in\MZ_{>0}$ such that $I$ has a reduced primary decomposition $I = I_1\cap I_2\cap\cdots\cap I_k$ which is unique up to reordering. In particular, for all $j= 1$ to $k$, $I_1\cap \cdots\cap I_{j-1}\cap I_{j+1}\cap\cdots\cap I_k \nsubseteq I_j$. Also, $V_{\MC} (I) = V_{\MC} (I_1) \cup V_{\MC} (I_2)\cup\cdots\cup V_{\MC} (I_k) = V_{\MC} (\sqrt{I_1}) \cup V_{\MC} (\sqrt{I_2})\cup\cdots\cup V_{\MC} (\sqrt{I_k})$.

It is enough to show that for all $j$, $V_{\MC} (I_j)\setminus V_{\MC}(x_1x_2\cdots x_s)$ is strongly dense in $V_{\MC}(I_j)$. Let $j\in\{1,2,\cdots,k\}$.

Suppose $V_{\MC} (I_j)\setminus V_{\MC}(x_1x_2\cdots x_s)$ is empty. Then $V_{\MC} (I_j)\subseteq V_{\MC}(x_1x_2\cdots x_s)$. It follows that there exists $l\in\MZ_{>0}$ such that $(x_1x_2\cdots x_s)^l\in I_j$. From the reduced primary decomposition of $I$, there exists $f$ such that $f\notin I_j$ and for all $i\neq j$, $f\in I_i$. Then $f\notin I$ but $f\cdot (x_1x_2\cdots x_s)^l\in I$, a contradiction since $I = I : (x_1x_2\cdots x_s)^{\infty}$.

Hence $V_{\MC} (I_j)\setminus V_{\MC}(x_1x_2\cdots x_s)$ is non-empty. Since $I_j$ is primary, $\sqrt{I_j}$ is prime and $V_{\MC}(I_j)$ is irreducible. Further, $V_{\MC} (I_j)\setminus V_{\MC}(x_1x_2\cdots x_s)$ is an open subvariety of $V_{\MC}(I_j)$. The lemma follows from~\cite[Theorem~1, p.~82]{Mumford1988RedBook}.
\end{proof}

We proceed by relating $V_{\MR}(\CE_G)\cap \MR^s_{\geq 0}$ with $V_{\MC}(\CE_G)$ and $V_{\MR}(\CE_G)\cap \MR^s_{> 0}$ with $V_{\MC} (\CE_G)\setminus V_{\MC}(x_1x_2\cdots x_s)$.

\begin{lemma}\label{lem:alg}
Let $G$ be a weakly-reversible, non-catalytic CRN. Then $V_{\MR}(\CE_G)\cap\MR^s_{>0}$ is strongly dense in $V_{\MR}(\CE_G)\cap \MR^s_{\geq 0}$.
\end{lemma}
\begin{proof}
 Let $\rho:V_{\MC}(\CE_G)\rightarrow \MR^s_{\geq 0}$ be given by $\langle z_1,z_2,\cdots,z_n\rangle\mapsto \langle|z_1|,|z_2|,\cdots,|z_n|\rangle$. Note that the image of $\rho$ is contained in $V_\MR(\CE_G)$ because $(\CE_G)$ is a pure difference binomial ideal (i.e., all reactions are reversible, and all specific rates are ``+1''). Suppose $\bf{z}\in V_{\MR}(\CE_G)\cap \MR^s_{\geq 0}$. We will find a sequence in $V_{\MR}(\CE_G)\cap \MR^s_{> 0}$ with limit $\bf{z}$. From Lemma~\ref{lem:saturated}, $(\CE_G)=(\CE_G):(x_1x_2\cdots x_s)^\infty$. From Lemma~\ref{lem:dense} applied to $(\CE_G)$, there exists a sequence $p_1,p_2,\cdots$ in $V_{\MC} (\CE_G)\setminus V_{\MC}(x_1x_2\cdots x_s)$ with limit $\bf{z}$. Consider the sequence $\rho(p_1), \rho(p_2), \cdots$. This sequence lies in $\MR^s_{>0}$ because each $p_i$ lies outside $V_\MC(x_1x_2\cdots x_s)$, and its limit is $\bf{z} = \rho(\bf{z})$.
\end{proof}

\begin{proof}[\textbf{Proof of Theorem~\ref{thm:norelsiph}}]
Let $G$ be a weakly-reversible, non-catalytic CRN, let $\CE=\CE_G$, and let $z\in \partial\MR^s_{\geq 0}$. We will show that $z$ is not a critical siphon-point. From Lemma~\ref{lem:relsiph}, this is enough. For all $r>0$, let $B_z(r) = \{x\in\MR^s_{\geq 0} : ||x-z|| \leq r\}$ be  the intersection with the non-negative orthant of the closed ball centered at $z$ and of radius $r$. For all $x\in\MR^s_{\geq 0}$, let $P_x$ denote the invariant polyhedron (of $U_G$) containing $x$. If $P_z$ does not intersect $\MR^s_{>0}$, we are done. Suppose $P_z$ intersects $\MR^s_{>0}$. From continuity, there exists a sufficiently small $r>0$ such that for every $x\in B_z(r)$, $P_x$ intersects $\MR^s_{>0}$. Let $r_0$ be such a sufficiently small $r$.

Note that the MAS $\langle U_G,1\rangle$ with the weight function $k=1$ is detailed balanced: the point of detailed balance corresponds to setting each species to have unit concentration. Therefore, for every $x\in B_z(r_0)$, from Birch's Theorem~\cite[Lemma~4B]{horn72general}, the set $P_x\cap V_{\MR}(\CE)\cap \MR^s_{> 0}$ of detailed balanced points in the relative interior of the invariant polyhedron contains exactly one point. For every $x\in B_z(r_0)$, let $d(x)$ be the distance of the unique point in $P_x\cap V_{\MR}(\CE)\cap \MR^s_{> 0}$ from $z$. Then the function $d:B_z(r_0)\rightarrow\MR_{\geq 0}$ which sends $x$ to $d(x)$ is a continuous function on a compact set, and hence attains its infimum. Let $x_0\in B_z(r_0)$ be a point where this infimum is attained.

If $d(x_0) = 0$ then $P_{x_0}$ is not an invariant polyhedron, a contradiction. Hence, $d(x_0) > 0$. Let $\epsilon = \min(r_0,\frac{d(x_0)}{2})$. The set $B_z(\epsilon)\cap V_{\MR}(\CE) \cap \MR^s_{>0}$ is empty by construction. From Lemma~\ref{lem:alg}, this implies that $z\notin V_{\MR}(\CE)$. Hence, by definition, $z$ is not a critical siphon-point.
\end{proof}

%There is also an alternate proof to Theorem~\ref{thm:norelsiph} that uses localization. We sketch the proof idea. I

\section{Atoms, Primes, and Catalysis}
The idea that all chemical species are uniquely divisible into ``atoms'' --- species that are immutable and indestructible --- dates back at least to the work of John Dalton in the early
19th century, and perhaps even further back to philosophical speculations in Indian, Greek, and Islamic civilizations. Mathematical abstractions of this idea in the setting of chemical reaction networks have recently been proposed independently by three groups: ``atomic event-systems''~\cite{adleman-2008} by Adleman \textit{et al.}; ``elemented'' and ``constructive networks''~\cite{DBLP:journals/siamam/ShinarAF09} by Shinar, Alon, and Feinberg; and ``normal networks'' and ``complete networks''~\cite{Gnacadja2009394} by Gnacadja. Here we present an algebraic formulation of these ideas that generalizes all three proposals. In~\cite{adleman-2008}, we have proved global stability results for detailed balanced atomic event-systems. Gnacadja has also proved global stability for complete networks~\cite{Gnacadja2009394}. In this section, we show that our main theorem that critical siphons occur only in catalytic networks (Theorem~\ref{thm:norelsiph}) generalizes these previous results. Specifically, we obtain persistence results for atomic event-systems, and normal networks~\cite{Gnacadja2009394} as easy corollaries to our main theorem. Constructive networks, on the other hand, may admit critical siphons, and proving they are persistent is not amenable to our techniques. Prime ideals will enter into our story, and one reason is this easy and well-known algebraic result.

%event syst: oct 2008
%normal networks: oct 2009
%constructive networks: march 2008, jan 2009

\begin{lemma}\label{lem:primesat}
If $I\in\mathbf{k}[x_1,x_2,\cdots,x_s]$ is a prime ideal then for all $f\notin I$, the ideals $I = I:f^\infty$.
\end{lemma}
\begin{proof}
It is clear that $I\subseteq I:f^\infty$. For the sake of contradiction, suppose $g\in I:f^\infty\setminus I$. Then there exists $k>0$ with $f^kg\in I$. Since $I$ is prime, and $g\notin I$, it follows that $f\in I$, a contradiction.
\end{proof}

\begin{theorem}\label{thm:primesat}
Let $G$ be a weakly-reversible CRN. If the ideal $(\CE_G)$ is prime then $G$ is not catalytic.
\end{theorem}
\begin{proof}
At the point $(1,1,...,1)$, every binomial in $\CE_G$ evaluates to zero, and hence every polynomial in $(\CE_G)$ evaluates to zero. But every monic monomial evaluates to $1$. In particular, the monic monomial $x_1x_2\cdots x_s$ does not belong to $(\CE_G)$. From Lemma~\ref{lem:primesat}, $(\CE) = (\CE):(x_1x_2\cdots x_s)^\infty$. From Lemma~\ref{lem:saturated}, $G$ is not catalytic.
\end{proof}

If $I$ is a binomial prime ideal, then the quotient ring $\mathbf{k}[x_1,x_2,\cdots,x_s]/I$ is isomorphic (though not canonically) to a ring of Laurent monomials. In chemical reaction networks, each such isomorphism corresponds to a decomposition map of species into ``pseudoatoms.'' We now make this precise.

The \textit{main component} of a pure difference binomial ideal $I\subseteq\MC[x_1,x_2,\cdots,x_s]$ is the minimal prime $P$ containing $I$ such that $V_{\MC}(I)\cap\MR^s_{>0}=V_{\MC}(P)\cap\MR^s_{>0}$. It is well-known that $P$ is also a pure difference binomial ideal: it is obtained as the lattice ideal of the saturation of the integer lattice corresponding to $I$~\cite{eisenbud1994binomial}. Therefore, $\mathbf{k}[x_1,x_2,\cdots,x_s]/P$ is isomorphic to a ring of Laurent monomials.

\begin{lemma}\label{lem:atomicdecomposition}
Let $m,s\in\MZ_{>0}$. For $i=1$ to $m$, let $p_i = (p_{i1},p_{i2},\cdots, p_{is})$ and $q_i = (q_{i1},q_{i2},\cdots,q_{is})$ be vectors in $\MZ^s_{\geq 0}$. Let $\CE =\{x^{p_i} - x^{q_i} \mid i=1,2,\cdots,m\}\subseteq \MC[x]$. Let $v_1,v_2,\cdots,v_l\in\MZ^s$ be a basis for $\{q_1-p_1,q_2-p_2,\cdots,q_m-p_m\}^\perp$ in $\MC^s$. Let $D_v: \MC[x]\rightarrow \MC[a_1,a_1^{-1},\cdots,a_l, a_l^{-1}]$ be the ring homomorphism that sends $x_i\mapsto a_1^{v_{1i}}a_2^{v_{2i}}\cdots a_l^{v_{li}}$. Then $\ker D_v$ equals the main component of $(\CE)$.
\end{lemma}

The proof is well-known in the literature, and details can be found, for example, in \cite{eisenbud1994binomial}. If $\CE = \CE_G$ and the decomposition does not use negative numbers ($v_1,v_2,\cdots,v_l\in\MZ^s_{\geq 0}$) then $G$ together with the choice of isomorphism corresponds to an elemented network. Further, if every ``pseudoatom'' $a_i$ is itself a species --- i.e., for each $a_i$, there exists $x_j$ such that $x_j\mapsto a_i$ --- then $G$ is constructive. For elemented and constructive networks, the pseudoatoms are referred to as ``elements.''

Finite, atomic event-systems were defined in~\cite[Definition~2.11]{adleman-2008}. Here we give a chemical reaction network version of the definition for convenience, which also generalizes the notion to weakly-reversible networks.

\begin{definition}
Let $G$ be a weakly-reversible CRN. The \textit{atoms} of $G$ are $A_G = \{x_i\mid x_i$ is an isolated node in $\overline {G}\}$. An \textit{atomic monomial} is a monic monomial all of whose prime factors (variables that appear in the monomial) come from $A_G$. The CRN $G$ is \textit{atomic} iff every connected component of $\overline{G}$ contains precisely one atomic monomial.
\end{definition}

It is immediate from definitions that pre-normal networks of Gnacadja~\cite[Definition~6.1]{Gnacadja2009394} are a special case of atomic networks. The definition of pre-complete networks~\cite[Definition~7.1]{Gnacadja2009394} intuitively says that all species can be decomposed into ``elements'' (our atoms) in a unique way that is preserved by reactions. Further, pre-complete networks are restricted to ``reversible binding reactions'' that are of the form $\sum a_i x_i\autorightleftharpoons{}{} x_j$, so that every species that is not an element can be decomposed further. This gives, for every non-element, a path in the event-graph to a monomial of elements. By uniqueness of the decomposition, it follows that pre-complete networks are atomic.

\begin{lemma}\label{atomicisprime}
Let $G$ be a CRN. If $G$ is atomic then $(\CE_G)$ is prime.
\end{lemma}
\begin{proof}
Let $\CE = \CE_G$. Let $A_G=\{a_1,a_2,\cdots,a_l\}\subseteq\{x_1,x_2,\cdots,x_s\}$ be the atoms. We claim that $\mathbf{k}[x]/(\CE) \cong \mathbf{k}[a_1,a_2,\cdots,a_l]$. Since the right-hand side is an integral domain, $(\CE)$ must be prime.

To prove the isomorphism, we give an explicit ``decomposition map'' $D$ from $\mathbf{k}[x]$ to $\mathbf{k}[a_1,a_2,\cdots,a_l]$: $x_i$ gets sent to the unique atomic monomial from $A_G$ that belongs to the connected component of $x_i$ in $\overline{G}$. By definition of $A_\CE$, the atoms get sent to themselves. The map is defined to be the identity on $\mathbf{k}$, and extended to polynomials to make it a ring homomorphism.

If $M-N\in\CE$ then $D(M-N) = 0$. This is because $M$ and $N$ are in the same connected component of $\overline{G}$, and hence $D(M)=D(N)$. It follows that $(\CE)\subseteq \ker{D}$.

To show $\ker{D}\subseteq(\CE)$, first note that $\ker{D}$ is generated by pure difference binomials~\cite[Theorem~7.3]{miller2005comb}. It is sufficient to show that for all pure differences monomial $M - N = \prod x_i^{c_i} - \prod x_i^{d_i}$, if $D(M-N) = 0$ then $M - N\in (\CE)$. We claim that $M$ and $D(M)$ are in the same connected component of $\overline{G}$. This is certainly true if $M = x_i$, by the definition of $D$. It now follows in general, by the definition of $\overline{G}$, and since $D$ is a ring homomorphism. Since $D(M) = D(N)$, it follows that the monomials $M$ and $N$ are in the same connected component of $\overline{G}$. Therefore, by Lemma~\ref{lem:evtgraph}, $M - N \in (\CE)$.
\end{proof}

From Theorem~\ref{thm:primesat}, atomic and pre-complete networks are non-catalytic, and hence persistent due to Corollary~\ref{cor:persistent}. In this sense, non-catalytic networks generalize the notion of atomicity.

Recall that every complete network is pre-complete and admits a Lyapunov function. Similarly, every finite, natural, atomic event-system is atomic, and admits a Lyapunov function. Therefore, we obtain the global attractor conjecture for finite, natural, atomic event-systems~\cite[Theorem~6.1]{adleman-2008}, and for complete networks~\cite[Theorem~8.3]{Gnacadja2009394}, as a corollary to the persistence of non-catalytic networks (Corollary~\ref{cor:persistent}).

However, our proofs do not extend to constructive networks. Consider the following example.

\begin{example}\label{ex:constructive}
Consider the reaction $x \autorightleftharpoons{}{} 2x$. By choosing the set of pseudoatoms (elements) to be the empty set, one gets a map $D:\MC[x]\rightarrow\MC$. The map $D$ sends $x$ to $1$, and the ring $\MC[x]$ to the ring $\MC$. The kernel is $(x-1)$, which is the main component of the ideal $(x-x^2)$. From the remark after Lemma~\ref{lem:atomicdecomposition}, $\langle G,D\rangle$ is elemented, with no elements. It trivially becomes a constructive network, since there are no elements (pseudoatoms). Correspondingly, $\{x\}$ is a critical siphon.
\end{example}

The proof of Lemma~\ref{atomicisprime} does not hold for constructive networks because $M$ and $D(M)$ may belong to different connected components. In the example above, $D(x)$ is $1$, but $1$ and $x$ are in different connected components of the event-graph. To insist that $M$ and $D(M)$ be in the same connected component is to insist that the CRN provide a ``decomposition path'' from compounds to atoms. Constructive networks admit critical siphons because their definition does not insist on the existence of such decomposition paths.

\section{Autocatalysis}
The converse of Theorem~\ref{thm:norelsiph} is false. There exist examples of catalytic weakly-reversible CRNs with no critical siphons, as we now show.

\begin{example}\label{example4}
Consider the reversible reactions
\begin{align*}
2x &\autorightleftharpoons{}{} 0\
\\x + y &\autorightleftharpoons{}{} y
\end{align*}
There is an edge from the monomial $xy$ to the monomial $y$. However, there is no path in the event-graph from $x$ to $1$. Hence, the network is catalytic. However, there are no critical siphons. The only siphon is $\{y\}$ and this is not critical because $y$ is dynamically invariant, \textit{i.e.}, along solution trajectories, its value remains unchanged.
\end{example}

It is possible to make the class of catalytic networks smaller to exclude networks like the one in Example~\ref{example4}.

\begin{definition}\label{def:almsat}
A weakly-reversible CRN $G$ is \textit{strictly catalytic} iff there exist path-connected $M,N\in V(\overline{G})$ such that for all integers $k>0$, $(M/\gcd(M,N))^k$ and $(N/\gcd(M,N))^k$ are not path-connected in $V(\overline{G})$.
\end{definition}

Note that every catalytic network is strictly catalytic, with $k=1$. Weakly-reversible networks with critical siphons must be strictly catalytic: Theorem~\ref{thm:norelsiph}, Corollary~\ref{cor:persistent}, and Corollary~\ref{cor:gas} hold with ``catalytic'' replaced by ``strictly catalytic.'' This is because binomials of the form $M^k - N^k$ can be factorized into the binomial $M-N$ and another factor which is non-vanishing on the positive orthant. The non-negative variety $V_{\MR}(\CE_G)\cap\MR^s_{\geq 0}$ depends only on the binomial $M-N$. Hence, the proofs of Lemma~\ref{lem:alg} and Theorem~\ref{thm:norelsiph} go through even when ``catalytic'' is replaced by ``strictly catalytic.''

The next example shows a weakly-reversible CRN that is strictly catalytic but has no critical siphons.

\begin{example}\label{example5}
Consider the reaction $x+a \autorightleftharpoons{}{} x+b$. This reaction is strictly catalytic since for all $k>0$, there is no path from $a^k$ to $b^k$. However, there are no critical siphons. This is because $x$ belongs to every siphon, and no siphon-point with value $0$ for $x$ is critical, because $x$ is dynamically invariant. Adding the reaction $x\autorightleftharpoons{}{} 2x$ introduces a stoichiometrically compatible direction transverse to the siphon, so that $x$ is no longer dynamically invariant, and the set $\{x\}$ becomes a critical siphon. This new network is also strictly catalytic.
\end{example}

The modification required of the network in Example~\ref{example5} was of a very special kind. It involved the introduction of a reaction where one molecule of $x$ can produce two molecules of $x$. In other words, the species $x$ was acting as an autocatalyst. We now make this precise.

\begin{definition}\label{def:autocatalytic}
A weakly-reversible CRN $G$ is \textit{autocatalytic} iff there exist $i\in [s]$, $a\in\MZ_{>0}$, $M,N\in \mathcal{M}_s$ such that
\begin{enumerate}
\item $x_i^a(M-x_iN)\in (\CE_G)$ and $x_i\nmid M$,
\item For all $k\in\MZ_{>0}$, $M^k -(x_iN)^k \notin (\CE_G)$.
\end{enumerate}
\end{definition}

Every autocatalytic network is catalytic, and in fact strictly catalytic. Note that the last CRN of Example~\ref{Seesaw1}, the CRN of Example~\ref{ex:constructive}, and the second CRN of Example~\ref{example5} are autocatalytic. Here is one more example of an autocatalytic network, where the autocatalysis involves two species.

\begin{example}
Consider the reactions
\begin{align*}
y &\autorightleftharpoons{}{} 2x\
\\2y &\autorightleftharpoons{}{} x
\end{align*}
The point $(0,0)$ is a critical siphon-point. The associated event-system is $\CE=\{y-x^2,y^2-x\}$. The ideal $(\CE)$ contains the binomials $y-y^4$ as well as $x-x^4$, however for all $k\in\MZ_{>0}$, the binomials $1 - x^{3k}$ and $1 - y^{3k}$ do not belong to $(\CE)$. Hence, the CRN is autocatalytic.
\end{example}

These examples prompt the following conjecture.

\begin{conjecture}\label{conj:auto}
Let $G$ be a weakly-reversible CRN. Then $G$ is autocatalytic iff $G$ admits a critical siphon.
\end{conjecture}

If the conjecture is true, one may hope to attack Feinberg's persistence conjecture (Open~\ref{open:feinberg}) by exploiting the special structure of autocatalytic networks.

\subsubsection*{Acknowledgments}
This work was prompted by discussions with Len Adleman. The ideas in the proofs of Lemma~\ref{lem:dense} and Lemma~\ref{lem:alg} are due to Najimuddin Fakhruddin. I thank Bernd Sturmfels, Dustin Reishus, Anne Shiu, and Ezra Miller for helpful discussions. Any mistakes that remain are mine alone.

\bibliographystyle{amsplain}
\bibliography{../eventsystems}

\providecommand{\bysame}{\leavevmode\hbox to3em{\hrulefill}\thinspace}
\providecommand{\MR}{\relax\ifhmode\unskip\space\fi MR }
\providecommand{\MRhref}[2]{%
  \href{http://www.ams.org/mathscinet-getitem?mr=#1}{#2}
}
\providecommand{\href}[2]{#2}
\begin{thebibliography}{10}

\bibitem{adleman-2008}
Leonard Adleman, Manoj Gopalkrishnan, Ming-Deh Huang, Pablo Moisset, and Dustin
  Reishus, \emph{On the mathematics of the law of mass action},
  \url{http://arXiv.org:0810.1108}, 2008.

\bibitem{anderson:1464}
David~F. Anderson, \emph{Global asymptotic stability for a class of nonlinear
  chemical equations}, SIAM J. Appl. Math. \textbf{68} (2008), no.~5,
  1464--1476.

\bibitem{anderson:1840}
David~F. Anderson and Anne Shiu, \emph{The dynamics of weakly reversible
  population processes near facets}, SIAM J. Appl. Math. \textbf{70} (2010),
  no.~6, 1840--1858.

\bibitem{Angeli2007598}
David Angeli, Patrick~De Leenheer, and Eduardo~D. Sontag, \emph{A {Petri} net
  approach to the study of persistence in chemical reaction networks}, Math.
  Biosci. \textbf{210} (2007), no.~2, 598--618.

\bibitem{Ceruzzi1981Zuse}
Paul~E. Ceruzzi, \emph{The early computers of {Konrad Zuse}, 1935 to 1945},
  IEEE Ann. Hist. Comput. \textbf{3} (1981), no.~3, 241--262.

\bibitem{Cox1991Ideals}
David Cox, John Little, and Donal O'Shea, \emph{Ideals, varieties and
  algorithms: An introduction to computational algebraic geometry and
  commutative algebra}, Springer, 1991.

\bibitem{Craciun20091551}
Gheorghe Craciun, Alicia Dickenstein, Anne Shiu, and Bernd Sturmfels,
  \emph{Toric dynamical systems}, Journal of Symbolic Computation \textbf{44}
  (2009), no.~11, 1551--1565, In Memoriam Karin Gatermann.

\bibitem{eisenbud1994binomial}
David Eisenbud and Bernd Sturmfels, \emph{Binomial ideals}, Duke Math. J.
  \textbf{84} (1996), no.~1, 1--45.

\bibitem{feinberg72complex}
Martin Feinberg, \emph{Complex balancing in general kinetic systems}, Arch.
  Rational Mech. Anal. \textbf{49} (1972).

\bibitem{feinberg72chemical}
\bysame, \emph{On chemical kinetics of a certain class}, Arch. Rational Mech.
  Anal. \textbf{46} (1972).

\bibitem{feinberg79lectures}
\bysame, \emph{Lectures on chemical reaction networks},
  \url{http://www.che.eng.ohio-state.edu/~FEINBERG/LecturesOnReactionNetworks/%
}, 1979.

\bibitem{Feinberg19872229}
\bysame, \emph{Chemical reaction network structure and the stability of complex
  isothermal reactors--i. the deficiency zero and deficiency one theorems},
  Chem. Eng. Sci. \textbf{42} (1987), no.~10, 2229--2268.

\bibitem{feinberg95existence}
\bysame, \emph{The existence and uniqueness of steady states for a class of
  chemical reaction networks}, Arch. Rational Mech. Anal. \textbf{132} (1995),
  311--370.

\bibitem{Gnacadja2010Reach}
Gilles Gnacadja, \emph{Reachability, persistence and constructive chemical
  networks}, preprint (2010), available at
  \texttt{http://math.gillesgnacadja.info/files/ConstructiveCRNT.html}.

\bibitem{Gnacadja2009394}
\bysame, \emph{Univalent positive polynomial maps and the equilibrium state of
  chemical networks of reversible binding reactions}, Adv. Appl. Math.
  \textbf{43} (2009), no.~4, 394--414.

\bibitem{gunawardena2003chemical}
Jeremy Gunawardena, \emph{Chemical reaction network theory for in-silico
  biologists}, \url{http://vcp.med.harvard.edu/papers/crnt.pdf}, 2003.

\bibitem{horn72necessary}
Friedrich J.~M. Horn, \emph{Necessary and sufficient conditions for complex
  balancing in chemical kinetics}, Arch. Rational Mech. Anal. \textbf{49}
  (1972).

\bibitem{horn74dynamics}
Friedrich J.~M. Horn, \emph{The dynamics of open reaction systems},
  Mathematical aspects of chemical and biochemical problems and quantum
  chemistry (New York), Proc. SIAM-AMS Sympos. Appl. Math., vol. VIII, 1974.

\bibitem{horn72general}
Friedrich J.~M. Horn and Roy Jackson, \emph{General mass action kinetics},
  Arch. Rational Mech. Anal. \textbf{49} (1972), 81--116.

\bibitem{Hungerford1974Algebra}
Thomas~W. Hungerford, \emph{Algebra}, Springer-Verlag, New York [N.Y.], 1980.

\bibitem{miller2005comb}
Ezra Miller and Bernd Sturmfels, \emph{Combinatorial commutative algebra},
  Springer Science+Business Media, Inc.

\bibitem{Mumford1988RedBook}
David Mumford, \emph{The red book of varieties and schemes}, Springer-Verlag,
  1988.

\bibitem{seesawgates}
Lulu Qian and Erik Winfree, \emph{{A simple DNA gate motif for synthesizing
  large-scale circuits}}, J. R. Soc. Interface (2011).

\bibitem{DBLP:journals/siamam/ShinarAF09}
Guy Shinar, Uri Alon, and Martin Feinberg, \emph{Sensitivity and robustness in
  chemical reaction networks}, SIAM J. Appl. Math. \textbf{69} (2009), no.~4,
  977--998.

\bibitem{Shiu2010Siphons}
Anne Shiu and Bernd Sturmfels, \emph{Siphons in chemical reaction networks},
  Bull. Math. Biol. (2010).

\bibitem{sontag01structure}
Eduardo~D. Sontag, \emph{Structure and stability of certain chemical networks
  and applications to the kinetic proofreading model of {T}-cell receptor
  signal transduction}, IEEE Trans. Autom. Control \textbf{46} (2001),
  1028--1047.

\bibitem{Zhang16112007}
David~Yu Zhang, Andrew~J. Turberfield, Bernard Yurke, and Erik Winfree,
  \emph{{Engineering Entropy-Driven Reactions and Networks Catalyzed by DNA}},
  Science \textbf{318} (2007), no.~5853, 1121--1125.

\end{thebibliography}
\end{document}